%% file: article.tex
\documentclass[10pt]{IEEEtran}
\usepackage{amsmath}
\usepackage{amssymb}
\usepackage{amsthm}
\usepackage{centernot}
\usepackage{colortbl}
\usepackage[T1]{fontenc}
\usepackage{multirow}
\usepackage{subcaption}
\usepackage{tikz}

\newtheorem{theorem}{Theorem}
\newtheorem{lemma}[theorem]{Lemma}
\newtheorem{proposition}[theorem]{Proposition}
\newtheorem{definition}[theorem]{Definition}
\newtheorem{assumption}[theorem]{Assumption}

\DeclareMathOperator{\cost}{c}
\DeclareMathOperator{\RI}{RI}
\DeclareMathOperator{\AU}{AU}

\newcommand{\RIO}[2]{[#1, #2)}

\usetikzlibrary{calc}
\usetikzlibrary{chains}
\usetikzlibrary{intersections}
\usetikzlibrary{positioning}
\usetikzlibrary{shapes}

\tikzstyle{node} = [circle, draw, minimum size = 0.7 cm,
                    inner sep = 0pt]

\newlength{\gd}
\newlength{\dan}
\newsavebox{\imageboxa}
\newsavebox{\imageboxb}
\newsavebox{\imageboxc}
\newsavebox{\imageboxd}

\title{Is dynamic dedicated path protection tractable?}

\author{Ireneusz Szcześniak, Ireneusz Olszewski, Bożena Woźna-Szcześniak

  \thanks{I.~Szcześniak is with the Department of Computer Science of
    the Częstochowa University of Technology, Częstochowa, Poland,
    iszczesniak@pcz.pl}%
  \thanks{I.~Olszewski is with the Faculty of Telecommunications,
    Computer Science and Electrical Engineering, Bydgoszcz University
    of Science and Technology, Bydgoszcz, Poland, irek@pbs.edu.pl}%
  \thanks{B.~Woźna-Szcześniak is with the Department of Mathematics
    and Computer Science of the Jan Długosz University, Częstochowa,
    Poland, b.wozna@ujd.edu.pl}}

\begin{document}

\maketitle

\begin{abstract}
Intractable is the problem of finding two link-disjoint paths of
minimal cost if the path cost is limited since it can be a special
case of the partition problem.  In optical networks, this limit can be
introduced by the signal modulation reach.  Even without this limit,
the existing literature suggested the problem intractable because of
the spectrum continuity and contiguity constraints, but we show that
the problem can be solved exactly with the recently-proposed
\emph{generic} Dijkstra algorithm over a polynomially-bounded search
space, thus proving the problem tractable.
\end{abstract}


\section{Introduction}


Routing is a combinatorial optimization problem where we search for an
exact (optimal) solution in a network.  Whether a specific problem is
tractable (of polynomial time) or intractable (usually of exponential
time) depends on, most notably, the constraints or the number of paths
to find.  Heuristic algorithms can solve an intractable problem
quicker than an exact algorithm but they do not guarantee exact
results.


The simplest is the shortest path problem solved by the celebrated
Dijkstra algorithm.  However, the Dijkstra algorithm is unable to find
a path in optical networks.  Optical spectrum is divided into a finite
set of fine frequency slot units, and a connection should use on all
links of its path the same units (continuity constraint) that are
contiguous (contiguity constraint).  The \emph{generic} Dijkstra
algorithm can do that \cite{10.1364/JOCN.11.000568}.


In \emph{static} routing a number of connections are routed in an
unloaded network, while in \emph{dynamic} routing a single connection
is routed in a loaded network.  Static (aka offline) routing is more
about planning a network, while dynamic (aka online) routing is more
about operating a network.  In optical networks, static routing is
intractable but dynamic is not \cite{10.1109/NOMS56928.2023.10154322}.


Dedicated path protection establishes a connection made of two
link-disjoint paths: the working and the protecting.  The problem of
finding two link-disjoint paths of minimal cost is tractable and
solved by the Suurballe algorithm \cite{10.1002/net.3230040204}.  For
optical networks, however, the Suurballe algorithm is unable to take
into account the continuity and contiguity constraints, and whether an
efficient and exact algorithm existed was unknown.


Our contribution is a proof that dynamic dedicated path protection
(DDPP) in optical networks is tractable if the costs of the working
and protecting paths are unlimited.

\section{Related works}


In \cite{10.1109/INFCOM.2004.1354524}, for DDPP in wavelength-division
multiplexed (WDM) networks, in order to minimize the number of links
used, the authors propose an integer-linear programming (ILP)
formulation and heuristic algorithms.  The authors also prove (Theorem
2) that a \emph{related decision problem} of whether there exists a
pair of link-disjoint paths (without requiring its cost be at most a
given value, thus making this problem different from the decision
problem of DDPP) meeting the continuity constraint is
non-deterministic polynomial-time complete (NP-complete) by reducing
from the 3SAT problem: for every instance of the 3SAT problem, there
exists a corresponding pair of link-disjoint paths (regardless of its
cost), and vice versa.  The number of instances of the 3SAT problem is
exponential, so exponential at least is the number of link-disjoint
path pairs, including the most expensive.  An easier problem of
finding the most expensive \emph{single path} has its decision problem
(that requires a path of at least a given cost) NP-complete.  The
authors argue that the proof validates the study of ILP formulations
and heuristics, suggesting that the DDPP problem is intractable.
However, DDPP searches for pairs not of any cost but of minimal cost
only, and so its complexity class remained unknown.


The problem of finding two link-disjoint paths becomes intractable
when we require the length of a path be at most $K$, as proven by
Theorem 3.6 in \cite{10.1002/net.3230120306}: the corresponding
decision problem is NP-complete since it becomes the partition problem
in a special case.  In optical networks, for any value of $K$ that
models the reach of the least spectrally-efficient modulation, DDPP
also becomes intractable \cite{10.1109/ACCESS.2019.2901018} because
the problem of finding two limited link-disjoint paths (without the
spectrum continuity and contiguity constraints) is a special case of
DDPP when an optical network has a single unit available on every
link.


In \cite{10.3390/e23091116}, we proposed an algorithm for DDPP that is
based on our generic Dijkstra algorithm, or more specifically, that is
the generic Dijkstra algorithm fed with the search graph purpose-built
for DDPP, generated on the fly from the input graph.  The simulation
experimental evidence suggested that the problem could be tractable
because we were able to find solutions in networks of a hundred nodes
and 640 units in reasonable time.  In contrast, intractable problems
can be often solved quickly for small problem sizes, but as the
problem size increases (even by one), the computation time increases
drastically, usually exponentially -- that was not pronounced in our
simulations.  While the signal modulation reach was not part of the
problem statement, in simulations it was, and most likely killed some
simulations with memory overuse, as reported in that work.
Admittedly, we did not know that introducing the limit would make that
big a difference.

\section{Preliminaries}



The generic Dijkstra algorithm operates on the search graph generated
from the input graph that models the optical network.  To
differentiate between the input and the search graphs, we call the
input graph the network (with nodes and links), and the search graph
just a graph (with vertexes and edges).  That split in terminology
clarifies presentation: a vertex in the input graph is just a node; a
vertex in the search graph is just a vertex.  A network models the
optical network, a graph models the search graph.


The network is made of $N$ nodes and directed links.  Parallel links
are allowed.  Link $k$ has cost $\cost(k)$ and the set $\AU(k)$ of
available units for establishing a connection.  We denote the set of
all units as $\Omega$.  Route $r$ is a sequence of neighboring links.
\emph{Trait} $t$ is a pair of cost $\cost(t)$ of total order $<$ and
units $\RI(t)$.  Trait provides information needed to establish a
connection along route $r$.  Units $\RI(t)$ is a half-closed integer
interval that models the contiguous units available along the route.
For example, $\RI(t) = \RIO{0}{2}$ means units 0 and 1.  Different
traits can refer to the same route if a connection can be established
using different contiguous units.


The graph is made of vertexes $(n_a, n_b)$, i.e., a pair of nodes,
that model a pair of routes: one that leads to node $n_a$, the other
to $n_b$.  Edge $e$ represents a way of transitioning from $(n_a,
n_b)$ to some other vertex by adding a link to one of the routes.  The
graph is generated on-the-fly, and we start searing at the vertex for
nodes where the routes begin.  For DDPP, the routes begin at the same
node $n_s$, so the first vertex is $(n_s, n_s)$.


Path $p$ is a sequence of neighboring edges.  \emph{Label} $l$
provides information needed to establish a connection along path $p$
made of two routes, and so a label is a pair of two traits.  Different
labels can refer to the same path.  In standard Dijkstra, a path and a
label are synonymous (i.e., they represent the same solution), since a
path can have a single label.  In generic Dijkstra, a path does not
represent a solution; it is a label that do, since a path can have
many labels.


Appending edge $e$ to path $p$ (Fig.~\ref{f:graph}) represents adding
a specific link $k$ to one of the routes.  There may be other links
leaving the end nodes of the two routes, but edges other than $e$
would represent those links.  Since dedicated protection offers
link-disjoint routes, we require no two edges in a path model the same
link.


Fig.~\ref{f:nets} shows the only three ways of appending link $k$ to
one of the two routes that ends at the source node of link $k$.  In
Figures \ref{f:network1} and \ref{f:network2}, route $r_a$ leads to
node $n_a$ with trait $t_a$, and route $r_b$ to node $n_b$ with trait
$t_b$, and so path $p$ reaches vertex $w = (n_a, n_b)$ with label $l =
(t_a, t_b)$.  Link $k$ can only be added to route $r_a$.  In
Fig.~\ref{f:network3}, both routes $r_a$ and $r'_a$ lead to node $n_a$
with traits $t_a$ and $t'_a$, and so path $p$ reaches vertex $w =
(n_a, n_a)$ with label $l = (t_a, t'_a)$.  Link $k$ can be added to
either of the routes because they both end at the same node.

\input{fig_graph.tikz}

\input{fig_nets.tikz}


Deriving labels $l'$ from label $l$, denoted by $l' \in l \oplus e$,
models appending edge $e$ to the path of label $l$.  Label $l$ and its
derived labels $l'$ differ by one trait only: the trait of the route
with a link appended.  For the cases in Figures \ref{f:network1} and
\ref{f:network2}, link $k$ is appended to route $r_a$ to yield
candidate labels as given by (\ref{e:candidate12}) for vertex $x =
(n_k, n_b)$ for Fig.~\ref{f:network1} and vertex $x = (n_b, n_b)$ for
Fig.~\ref{f:network2}.  For the case in Fig.~\ref{f:network3}, link
$k$ is appended to either of the routes $r_a$ and $r_a'$ to yield
candidate labels as given by (\ref{e:candidate3}) for vertex $x =
(n_a, n_k)$.

\begin{equation}
  l' \in l \oplus e = \{(t, t_b): t \in t_a \oplus k\}
  \label{e:candidate12}
\end{equation}

\begin{equation}
  \begin{split}
    l' \in l \oplus e =
    &\{(t, t_a'): t \in t_a \oplus k\}\text{ }\cup\\
    &\{(t_a, t): t \in t_a' \oplus k\}
  \end{split}
  \label{e:candidate3}
\end{equation}


Deriving traits $t'$ from trait $t$, denoted by $t' \in t \oplus k$,
models appending link $k$ to the route of trait $t$.  How a trait is
defined and derived depends on the network characteristics being
modeled; we only require Assumption \ref{a:trait} holds, where the
trait efficiency is given by Definition \ref{d:trait_efficiency}, and
$\preceq$ is the \emph{better than or equal to} relation.  For the
trait in optical networks, $\preceq$ is given by (\ref{e:preceqtrait})
and $\oplus$ by (\ref{e:oplustrait}), for which Assumption
\ref{a:trait} holds if $\oplus$ for cost is addition.

\begin{assumption}
  An inefficient trait yields inefficient traits:\\
  $t_i \preceq t_j \Rightarrow \forall t \in t_j \oplus k \; \exists
  t' \in t_i \oplus k: t' \preceq t$.
  \label{a:trait}
\end{assumption}

\begin{definition}[Trait efficiency]
Trait $t$ is efficient if, for a given node, there does not exist
trait $t'$ such that $t' \preceq t$.
\label{d:trait_efficiency}
\end{definition}

\begin{equation}
  t_i \preceq t_j \Leftrightarrow \cost(t_i) \le \cost(t_j) \text{ and }
  \RI(t_i) \supseteq \RI(t_j)
  \label{e:preceqtrait}
\end{equation}

\begin{equation}
  \begin{split}
    t' \in t \oplus k =
    \{(c, R): \text{ } & c = \cost(t) \oplus \cost(k) \text{ and }\\
    & R = \RI(t) \cap \AU(k)\}
  \end{split}
  \label{e:oplustrait}
\end{equation}


The \emph{Bellman principle of optimality} stipulates the optimal
substructure of a routing problem, known as the shortest-path tree
\cite{bellman}.  The \emph{generic principle of optimality} stipulates
the efficient substructure, termed the efficient-path tree in
\cite{10.1109/NOMS56928.2023.10154322}, but here we call it, more
precisely, the efficient-label tree.  The Bellman principle allows for
a single optimal label because ordering between labels is total;
suboptimal labels are discarded by relaxation as they are assumed
unable to yield optimal labels.  The generic principle allows for a
set of efficient traits (labels) because their ordering is partial;
inefficient traits (labels) are discarded by relaxation as they are
unable to yield efficient traits by Assumption \ref{a:trait} (labels
by Propositions \ref{p:label} and \ref{p:labelprime}).


Label efficiency is defined in terms of some relation $R$, as given by
Definition \ref{d:label_efficiency}.  We introduce relation $\preceq$
for labels, and call some labels $\preceq$-efficient.  Relation
$\preceq$ compares labels differently depending on whether the end
nodes of their routes are different or the same.  For different end
nodes, label $l_i = (t_{i, a}, t_{i, b})$ is better than or equal to
$l_j = (t_{j, a}, t_{j, b})$, denoted by $\preceq_{\ne}$, if the
traits of $l_i$ are better than or equal to the traits of $l_j$, as
given by (\ref{e:cmp!=}), otherwise the labels are
$\preceq_=$-incomparable, denoted by $\parallel$ in Table
\ref{t:l_rels}. For the same end nodes, the relation is $\preceq_=$.


\begin{definition}[Label R-efficiency]
Label $l$ is R-efficient if, for a given vertex, $\nexists l': l' R l$.
\label{d:label_efficiency}
\end{definition}


\input{tab_rels.tab}

\begin{equation}
  l_i \preceq_{\ne} l_j \Leftrightarrow t_{i, a} \preceq t_{j, a}
  \text{ and } t_{i, b} \preceq t_{j, b}
  \label{e:cmp!=}
\end{equation}

\section{A shortcoming and its correction}

In \cite{10.3390/e23091116}, for labels $l_i = (t_i, t'_i)$ and $l_j =
(t_j, t'_j)$ whose both routes terminate at the same node, the better
than or equal to relation (denoted there with $\le$) was flawed.

For sorted labels (i.e., with their traits sorted: $t_i \preceq t'_i$
and $t_j \preceq t'_j$), label $l_i$ was considered better than or
equal to $l_j$ if and only if the normal comparison $l_i \preceq_n
l_j$ held, as given by (\ref{e:normal}).  The cross comparison $l_i
\preceq_x l_j$, as given by (\ref{e:cross}), used with the normal
comparison (i.e., $l_i \preceq_n l_j$ or $l_i \preceq_x l_j$) would be
redundant because $\preceq_x$ would not hold if $\preceq_n$ did not
(Proposition \ref{t:redundant}).  Using $\preceq_x$ only would be
incorrect, because even if $\preceq_x$ does not hold, $\preceq_n$ can
hold (Proposition \ref{t:incorrect}).  All that is correct.

\begin{equation}
  l_i \preceq_n l_j \Leftrightarrow t_i \preceq t_j \text{ and } t'_i
  \preceq t'_j
  \label{e:normal}
\end{equation}

\begin{equation}
  l_i \preceq_x l_j \Leftrightarrow t_i \preceq t'_j \text{ and } t'_i
  \preceq t_j
  \label{e:cross}
\end{equation}


Incorrect, however, was using only $\preceq_n$ for unsorted labels
(i.e., when at least one label is unsorted).  In
\cite{10.3390/e23091116}, we assumed that we could use
(\ref{e:normal}) if for an unsorted label we swapped its traits, but
we were wrong as labels with incomparable traits cannot be sorted.
That is a shortcoming.


A correction: the normal or the cross comparison must hold, as given
by (\ref{e:cmp=}) and denoted by $\preceq_=$; no simpler comparison
exists for any labels (sorted or not) using $\preceq$ for traits.
There are only three comparisons possible: $\preceq_n$, $\preceq_x$
and $\preceq_=$.  Proposition \ref{t:redundant} for unsorted labels
does not hold so $\preceq_n$ cannot be used alone, and by Proposition
\ref{t:incorrect} comparison $\preceq_x$ cannot be used alone at all
(even for sorted labels).  Only (\ref{e:cmp=}) remains, provided we
can only use $\preceq$ for traits.

\begin{equation}
  l_i \preceq_= l_j \Leftrightarrow l_i \preceq_n l_j \text{ or } l_i
  \preceq_x l_j
  \label{e:cmp=}
\end{equation}

\section{Problem statement and resolution}


Prove that the generic Dijkstra algorithm can find two link-disjoint
routes of minimal cost that meet the resource continuity and
contiguity constrains.  Furthermore, show that the problem is
tractable if the route cost is unlimited.


We rely on the proof of correctness of generic Dijkstra in
\cite{10.1109/NOMS56928.2023.10154322}: that proof holds if the
problem has the efficient substructure, and that we claim for DDPP in
Lemma \ref{l:efficient}.  Furthermore, we show that $\preceq$ leads to
exponential complexity, and therefore introduce a new relation
$\preceq'$ that polynomially-binds the size of the search space.


In Propositions \ref{p:case1}, \ref{p:case2} and \ref{p:case3} for
$\preceq$ (and then in \ref{p:case1prime}, \ref{p:case2prime} and
\ref{p:case3prime} for $\preceq'$) we consider the only three cases
shown in Fig.~\ref{f:nets} that are modeled as in
Fig.~\ref{f:substructure} where two paths lead to vertex $w$: $p_i$
with efficient label $l_i$, and $p_j$ with inefficient label $l_j$.
Label $l_i$ yields candidate labels $l' \in l_i \oplus e$ when edge
$e$ is appended to path $p_i$, and likewise $l \in l_j \oplus e$ for
path $p_j$.  In the proofs, for brevity, we consider the link $k$
being added to the first route only; analogous reasoning holds for the
other route.  We show that for every label $l$, there exists $l'$ such
that $l'$ is better than or equal to $l$; the relation does not hold
for every $l'$ and $l$.

\subsection{Ordering $\preceq$}

\input{fig_substructure.tikz}

\begin{proposition}
  $l_i \preceq_{\ne} l_j \Rightarrow \forall l \in l_j \oplus e \;
  \exists l' \in l_i \oplus e: l' \preceq_{\ne} l$
  \label{p:case1}
\end{proposition}

\begin{proof}
  For $l_i = (t_{i, a}, t_{i, b})$, $l_j = (t_{j, a}, t_{j, b})$,
  derived labels are $l' \in \{(t', t_{i, b}): t' \in t_{i, a} \oplus
  k\}$ and $l \in \{(t, t_{j, b}): t \in t_{j, a} \oplus k\}$.  Since
  $l_i \preceq_{\ne} l_j$, relations $t_{i, a} \preceq t_{j, a}$ and
  $t_{i, b} \preceq t_{j, b}$ hold.  Since $t_{i, a} \preceq t_{j,
    a}$, for every $t$ there is $t'$ such that $t' \preceq t$ by
  Assumption \ref{a:trait}, so for every $l$ there is $l'$ such that
  $l' \preceq_{\ne} l$.
\end{proof}

\begin{proposition}
  $l_i \preceq_{\ne} l_j \Rightarrow \forall l \in l_j \oplus e \;
  \exists l' \in l_i \oplus e: l' \preceq_= l$
  \label{p:case2}
\end{proposition}

\begin{proof}
  For $l_i = (t_{i, a}, t_{i, b})$, $l_j = (t_{j, a}, t_{j, b})$,
  derived labels are $l' \in \{(t', t_{i, b}): t' \in t_{i, a} \oplus
  k\}$, and $l \in \{(t, t_{j, b}): t \in t_{j, a} \oplus k\}$.  Since
  $l_i \preceq_{\ne} l_j$, relations $t_{i, a} \preceq t_{j, a}$ and
  $t_{i, b} \preceq t_{j, b}$ hold.  Since $t_{i, a} \preceq t_{j,
    a}$, for every $t$, there exists $t'$ such that $t' \preceq t$ by
  Assumption \ref{a:trait}, so for every $l$ there exists $l'$ such
  that $l' \preceq_n l$.
\end{proof}

\begin{proposition}
  $l_i \preceq_= l_j \Rightarrow \forall l \in l_j \oplus e \;
  \exists l' \in l_i \oplus e: l' \preceq_{\ne} l$
  \label{p:case3}
\end{proposition}

\begin{proof}
  For labels $l_i = (t_i, t'_i)$, $l_j = (t_j, t'_j)$, derived labels
  are $l' \in \{(t', t'_i): t' \in t_i \oplus k\}$, and $l \in \{(t,
  t'_j): t \in t_j \oplus k\}$.  We assume $\preceq_=$ holds because
  $\preceq_n$ does; analogous reasoning holds for $\preceq_x$.  Since
  $l_i \preceq_n l_j$, relations $t_i \preceq t_j$ and $t'_i \preceq
  t'_j$ hold.  Since $t_i \preceq t_j$, for every $t$, there exists
  $t'$ such that $t' \preceq t$ by Assumption \ref{a:trait}, so for
  every $l$ there exists $l'$ such that $l' \preceq_{\ne} l$.
\end{proof}

\begin{proposition}
  An $\preceq$-inefficient label yields $\preceq$-inefficient labels:
  $l_i \preceq l_j \Rightarrow \forall l \in l_j \oplus e \; \exists
  l' \in l_i \oplus e: l' \preceq l$.
  \label{p:label}
\end{proposition}

\begin{proof}
  Follows from Propositions \ref{p:case1}, \ref{p:case2}, and
  \ref{p:case3}.
\end{proof}


The problem is that $\preceq_=$ leads to exponential worst-case
complexity.  In Fig.~\ref{f:lobe}, two routes from node $n_s$ to node
$n_x$ meet at $m$ intermediate nodes, which corresponds to finding a
path in a graph from vertex $s$ to vertex $x$.  The upper subroutes
have cost 0, the lower the consecutive powers of 2.  Labels $l_i$ and
$l_j$ are $\preceq_=$-incomparable if neither $l_i \preceq_= l_j$ nor
$l_i \succeq_= l_j$ holds.  Labels $l_i$ and $l_j$ are
$\preceq_=$-incomparable if one trait of $l_i$ is better than one
trait of $l_j$ while the other trait of $l_i$ is worse than the other
trait of $l_j$.  Since relation $\prec$ considers a trait of lower
cost better, then the number of $\preceq_=$-incomparable labels vertex
$x$ can have is $2^m$, i.e., labels of costs $(0, 2 ^ {(m + 1)} - 1)$,
$(1, 2 ^ {(m + 1)} - 2)$, \ldots, $(2^m - 1, 2^m)$, just as in the
partition problem.  In the worst case, routes would go through all
possible intermediate nodes and the subroutes would have a single
link.  The generic Dijkstra algorithm would keep these incomparable
labels and would find a correct path even when the route cost is
limited, albeit at the exponential complexity.

\input{fig_lobe.tikz}


\subsubsection{Example}

We demonstrate the efficient substructure and the worst-case
complexity of $\preceq_=$ with an example optical network in
Fig.~\ref{f:en}, where links are labeled with their name, cost and
available units.  We search for a pair of link-disjoint routes of
minimal cost meeting the spectrum constraints from $n_1$ to $n_3$.
The corresponding graph is shown in Fig.~\ref{f:eg}, where edges are
labeled with their name and the link being appended.


Table \ref{t:el1} lists the considered labels (both the efficient and
inefficient), their vertex, what they were yielded by, and the reason
for dropping the inefficient labels.  Fig.~\ref{f:et1} represents the
$\preceq$-efficient labels as a tree.  For example, there are two
labels for vertex $(n_1, n_2)$ yielded by label $l_1$: label $l_2$ for
edge $e_1$, and $l_3$ for $e_2$.  Label $l_2$ is $\preceq$-efficient
(and forms the efficient-label tree) but label $l_3$ is not because
$l_2 \preceq l_3$.


Since label $l_3$ is $\preceq_{\ne}$-inefficient, it can only yield
$\preceq$-inefficient labels.  Labels yielded for vertex $(n_1, n_3)$:
$l_4 \preceq_{\ne} ((0, \Omega), (1, \RIO{5}{7})) \in l_3 \oplus e_3$
and $l_5 \preceq_{\ne} ((0, \Omega), (3, \RIO{0}{2})) \in l_3 \oplus
e_4$, as by Proposition \ref{p:case1}.  A label yielded for vertex
$(n_2, n_2)$: $l_7 \preceq_= ((0, \RIO{0}{9}), (1, \RIO{0}{9})) \in
l_3 \oplus e_5$, as by Proposition \ref{p:case2}.  We could take the
previous $\preceq_=$-inefficient label yielded by $l_3 \oplus e_5$
(i.e., $((0, \RIO{0}{9}), (1, \RIO{0}{9}))$) to yield further
$\preceq_{\ne}$-inefficient labels for vertex $(n_2, n_3)$, as by
Proposition \ref{p:case3}.


For vertex $(n_3, n_3)$, having labels $l_{15}$ and $l_{17}$, we can
limit a route cost to 2, and find the required label: $l_{15}$.
However, if the cost of a route is unlimited, using relation
$\preceq_=$ is overkill because only one of these two labels should be
kept (either $l_{15}$ or $l_{17}$).

\input{fig_example_network.tikz}

\input{fig_example_graph.tikz}

\input{tab_labels1.tab}

\input{fig_example_tree1.tikz}

\subsection{Ordering $\preceq'$}


Luckily, for optical networks with unlimited route cost we can
introduce a new relation $\preceq'$ that does not cause exponential
complexity: $\preceq'_{\ne}$ and $\preceq'_=$ for labels whose routes
end at different or the same nodes, respectively.  When comparing
labels, trait costs can be considered equivalent if the label costs
are equal.  For instance, a label cost can be a sum of its trait
costs.  Relation $\prec'$ should find better a label of lower cost and
better resources, which combined with the equivalence relation (so
that equivalent labels can be discarded) is given by
(\ref{e:cmp!=prime}) and (\ref{e:cmp=prime}).  Resource $\RI(l)$ of
label $l$ is a pair of trait resources.  Resource inclusion relations
are given by (\ref{e:RI!=}-\ref{e:RIcross}).


For the cost $\cost(l)$ of label $l$, Assumption \ref{a:cost} should
hold.  It does for optical networks, where a label cost can be the sum
of trait costs and a trait cost is the route length which in turn is
the sum of link costs.  A trait cost can also be the product of route
length and modulation coefficient required for that length, since the
coefficient increases with length.

\begin{assumption}
  A label of higher or equal cost yields labels of higher or equal
  cost: $\cost(l_i) \le \cost(l_j) \Rightarrow \forall l' \in l_i
  \oplus e \; \forall l \in t_j \oplus e:\cost(l') \le \cost(l)$.
  \label{a:cost}
\end{assumption}

\begin{equation}
  l_i \preceq'_{\ne} l_j \Leftrightarrow \cost(l_i) \le \cost(l_j)
  \text{ and } \RI(l_i) \supseteq_{\ne} \RI(l_j)
  \label{e:cmp!=prime}
\end{equation}

\begin{equation}
  l_i \preceq'_= l_j \Leftrightarrow \cost(l_i) \le \cost(l_j) \text{
    and } \RI(l_i) \supseteq_= \RI(l_j)
  \label{e:cmp=prime}
\end{equation}

\begin{equation}
  \begin{split}
    \RI(l_i) \supseteq_{\ne} \RI(l_j) \Leftrightarrow
    &\RI(t_{i, a}) \supseteq \RI(t_{j, a}) \text{ and }\\
    &\RI(t_{i, b}) \supseteq \RI(t_{j, b})
  \end{split}
  \label{e:RI!=}
\end{equation}

\begin{equation}
  \begin{split}
    \RI(l_i) \supseteq_= \RI(l_j) \Leftrightarrow
    &\RI(l_i) \supseteq_n \RI(l_j) \text{ or }\\
    &\RI(l_i) \supseteq_x \RI(l_j)
  \end{split}
  \label{e:RI=}
\end{equation}

\begin{equation}
  \begin{split}
    \RI(l_i) \supseteq_n \RI(l_j) \Leftrightarrow
    &\RI(t_i) \supseteq \RI(t_j) \text{ and }\\
    &\RI(t'_i) \supseteq \RI(t'_j)
  \end{split}
  \label{e:RInormal}
\end{equation}

\begin{equation}
  \begin{split}
    \RI(l_i) \supseteq_x \RI(l_j) \Leftrightarrow
    &\RI(t_i) \supseteq \RI(t'_j) \text{ and }\\
    &\RI(t'_i) \supseteq \RI(t_j)
  \end{split}
  \label{e:RIcross}
\end{equation}

\begin{proposition}
  $l_i \preceq'_{\ne} l_j \Rightarrow \forall l \in l_j \oplus e \;
  \exists l' \in l_i \oplus e: l' \preceq'_{\ne} l$
  \label{p:case1prime}
\end{proposition}

\begin{proof}
  For $l_i = (t_{i, a}, t_{i, b})$, $l_j = (t_{j, a}, t_{j, b})$,
  derived labels are $l' \in \{(t', t_{i, b}): t' \in t_{i, a} \oplus
  k\}$, and $l \in \{(t, t_{j, b}): t \in t_{j, a} \oplus k\}$.  For
  $l' \preceq'_{\ne} l$ to hold, $\cost(l') \le \cost(l)$ and $\RI(l')
  \supseteq_{\ne} \RI(l)$ have to hold.  By Assumption \ref{a:cost},
  $\cost(l') \le \cost(l)$ holds.  Since $l_i \preceq'_{\ne} l_j$,
  relations $\RI(t_{i, a}) \supseteq \RI(t_{j, a})$ and $\RI(t_{i, b})
  \supseteq \RI(t_{j, b})$ hold.  For every $t \in t_{j, a} \oplus k$,
  there exists $t' \in t_{i, a} \oplus k$ such that $\RI(t') \supseteq
  \RI(t)$ because $\RI(t') \subseteq \RI(t_{i, a})$, $\RI(t) \subseteq
  \RI(t_{j, a})$ and $\RI(t_{i, a}) \supseteq \RI(t_{j, a})$, and thus
  $\RI(l') \supseteq_{\ne} \RI(l)$ holds.
\end{proof}

\begin{proposition}
  $l_i \preceq'_{\ne} l_j \Rightarrow \forall l \in l_j \oplus e \;
  \exists l' \in l_i \oplus e: l' \preceq'_= l$
  \label{p:case2prime}
\end{proposition}

\begin{proof}
  For $l_i = (t_{i, a}, t_{i, b})$, $l_j = (t_{j, a}, t_{j, b})$,
  derived labels are $l' \in \{(t', t_{i, b}): t' \in t_{i, a} \oplus
  k\}$, and $l \in \{(t, t_{j, b}): t \in t_{j, a} \oplus k\}$.  For
  $l' \preceq'_= l$ to hold, $\cost(l') \le \cost(l)$ and $\RI(l')
  \supseteq_= \RI(l)$ have to hold.  By Assumption \ref{a:cost},
  $\cost(l') \le \cost(l)$ holds.  Since $l_i \preceq'_{\ne} l_j$,
  relations $\RI(t_{i, a}) \supseteq \RI(t_{j, a})$ and $\RI(t_{i, b})
  \supseteq \RI(t_{j, b})$ hold.  For every $t \in t_{j, a} \oplus k$,
  there exists $t' \in t_{i, a} \oplus k$ such that $\RI(t') \supseteq
  \RI(t)$ because $\RI(t') \subseteq \RI(t_{i, a})$, $\RI(t) \subseteq
  \RI(t_{j, a})$ and $\RI(t_{i, a}) \supseteq \RI(t_{j, a})$, and thus
  $\RI(l') \supseteq_= \RI(l)$ holds.
\end{proof}

\begin{proposition}
  $l_i \preceq'_= l_j \Rightarrow \forall l \in l_j \oplus e \;
  \exists l' \in l_i \oplus e: l' \preceq'_{\ne} l$
  \label{p:case3prime}
\end{proposition}

\begin{proof}
  For $l_i = (t_i, t'_i)$, $l_j = (t_j, t'_j)$, derived labels are $l'
  \in \{(t', t'_i): t' \in t_i \oplus k\}$, and $l \in \{(t, t'_j): t
  \in t_j \oplus k\}$.  We assume $\preceq'_=$ holds because
  $\supseteq_n$ does; analogous reasoning holds for $\supseteq_x$.
  For $l' \preceq'_{\ne} l$ to hold, $\cost(l') \le \cost(l)$ and
  $\RI(l') \supseteq_{\ne} \RI(l)$ have to hold.  By Assumption
  \ref{a:cost}, $\cost(l') \le \cost(l)$ holds.  Since $l_i \preceq'_=
  l_j$, relations $\RI(t_i) \supseteq \RI(t_j)$ and $\RI(t'_i)
  \supseteq \RI(t'_j)$ hold.  For every $t \in t_j \oplus k$, there
  exists $t' \in t_i \oplus k$ such that $\RI(t') \supseteq \RI(t)$
  because $\RI(t') \subseteq \RI(t_{i, a})$, $\RI(t) \subseteq
  \RI(t_{j, a})$ and $\RI(t_{i, a}) \supseteq \RI(t_{j, a})$, and thus
  $\RI(l') \supseteq_{\ne} \RI(l)$ holds.
\end{proof}

\begin{proposition}
  An $\preceq'$-inefficient label yields $\preceq'$-inefficient
  labels: $l_i \preceq' l_j \Rightarrow \forall l \in l_j \oplus e
  \; \exists l' \in l_i \oplus e: l' \preceq' l$.
  \label{p:labelprime}
\end{proposition}

\begin{proof}
  Follows from Propositions \ref{p:case1prime}, \ref{p:case2prime},
  and \ref{p:case3prime}.
\end{proof}

\subsubsection{Example}


We revisit the example of the optical network in Fig.~\ref{f:en} and
use $\preceq'$ instead of $\preceq$.  Table \ref{t:el2} lists the 19
considered labels that are the subset of the 23 labels of Table
\ref{t:el1}.  We keep the label numbering the same as in the previous
example to allow for comparison.  Fig.~\ref{f:et2} shows the
$\preceq'$-efficient-label tree.


Label $l_{11}$ is $\preceq'$-inefficient ($l_9 \preceq' l_{11}$) and
cannot yield $\preceq'$-efficient labels by Proposition
\ref{p:labelprime}.  Therefore it is dropped and does not yield label
$l_{15}$, leaving only $l_{17}$ for vertex $(n_3, n_3)$.  Furthermore,
label $l_{20}$ is still dropped, but this time because $l_{17}
\preceq' l_{20}$.

\input{tab_labels2.tab}

\input{fig_example_tree2.tikz}

\begin{lemma}
  An $\preceq'$-efficient label has the $\preceq'$-efficient
  substructure, i.e., has been derived from an $\preceq'$-efficient
  label.
  \label{l:efficient}
\end{lemma}

\begin{proof}
  We prove by contradiction: we assume a label is
  $\preceq'$-efficient, and suppose that its substructure is not to
  conclude that the assumption was false (i.e., the label is
  $\preceq'$-inefficient).

  We search for $\preceq'$-efficient labels from vertex $s$, as shown
  in Fig.~\ref{f:substructure}.  To examine the label substructure, we
  assume a path to $x$ is made of a path (of at least one edge) to
  preceding vertex $w \ne s$ and the appended edge $e$.  We assume two
  paths lead to $w$: path $p_i$ with $\preceq'$-efficient label $l_i$,
  and path $p_j$ with $\preceq'$-inefficient label $l_j$, and so $l_i
  \preceq' l_j$.  The label of a path to $x$ is derived from a label
  of a path to $w$.

  We assume there exists path $p$ to vertex $x$ with
  $\preceq'$-efficient label $l$.  But we suppose an
  $\preceq'$-inefficient substructure: path $p$ is made of path $p_j$
  with $\preceq'$-inefficient label $l_j$.

  By Proposition \ref{p:labelprime}, for every label $l \in l_j \oplus
  e$ there exists $l' \in l_i \oplus e$ such that $l' \preceq' l$, and
  so $l$ cannot be $\preceq'$-efficient by Definition
  \ref{d:label_efficiency}, and that contradicts the assumption.  The
  proof applies to all preceding labels (of intermediate vertexes $w$
  of a path) as we track them back recursively to finally stop at the
  initial label.
\end{proof}

\subsection{Complexity}


For the worst case in Fig.~\ref{f:lobe}, if label cost is the sum of
trait costs (as usually is for DDPP), only one out of $2^m$ labels can
be kept for vertex $x$ because they all are of cost $2 ^ {(m + 1)} -
1$.  For optical networks, the number of $\preceq'$-incomparable
labels a vertex can have depends also on the number of units and is
$O(N^2|\Omega|^4)$ \cite{10.3390/e23091116}, a polynomial bound.
Since the size of the search space is polynomially bounded and generic
Dijkstra can find an exact solution, the DDPP problem with unlimited
route cost is tractable.

\section{Conclusion}


We discerned where the NP-completeness of dynamic dedicated path
protection in the optical networks comes from: the limit on the route
cost.  We dropped $\preceq$ that causes exponential complexity and
introduced a new relation $\preceq'$.


While the spectrum constraints make the problem harder, it is the
contiguity constraint (i.e., only contiguous units are allowed) that
polynomially binds the size of search space.


In optical networks, the generic Dijkstra algorithm can solve exactly
the unlimited problem, provided Assumption \ref{a:cost} holds and the
new relation $\preceq'$ is used.  We admit that using $\preceq$, as in
\cite{10.3390/e23091116}, would make the problem intractable.


The generic Dijkstra algorithm can also exactly solve the limited
problem, provided Assumption \ref{a:trait} holds and the relation
$\preceq$ is used because Lemma \ref{l:efficient} also holds for
$\preceq$ (since Proposition \ref{p:label} does).  With exponential
memory and time complexity.


The conclusions hold not only the link-disjoint routes that start and
end at the same nodes, but for any link-disjoint routes of minimal
cost that start and end at different nodes.

\bibliographystyle{IEEEtran}
\bibliography{all}

\section{Appendix}
\label{s:appendix}

\begin{proposition}
  For sorted labels $l_i$ and $l_j$ whose both routes terminate at the
  same node, $l_i \npreceq_n l_j \Rightarrow l_i \npreceq_x l_j$.
  \label{t:redundant}
\end{proposition}

\begin{proof}
  We prove the contrapositive: $l_i \preceq_x l_j \Rightarrow l_i
  \preceq_n l_j$.

  For labels $l_i = (t_i, t'_i)$ and $l_j = (t_j, t'_j)$ that are
  sorted (i.e., $t_i \preceq t'_i$ and $t_j \preceq t'_j$), if $l_i
  \preceq_x l_j$ holds (specifically, $t'_i \preceq t_j$ holds), then
  both $t_i \preceq t_j$ and $t'_i \preceq t'_j$ hold (because both
  $t_i \preceq t'_i \preceq t_j$ and $t'_i \preceq t_j \preceq t'_j$
  hold, and relation $\preceq$ is transitive), i.e., $l_i \preceq_n
  l_j$ holds.
\end{proof}

\begin{proposition}
  For sorted labels $l_i$ and $l_j$ whose both routes terminate at the
  same node, $l_i \npreceq_x l_j \centernot\Rightarrow l_i \npreceq_n
  l_j$.
  \label{t:incorrect}
\end{proposition}

\begin{proof}
  We prove the contrapositive: $l_i \preceq_n l_j
  \centernot\Rightarrow l_i \preceq_x l_j$.

  For labels $l_i = (t_i, t'_i)$ and $l_j = (t_j, t'_j)$ that are
  sorted (i.e., $t_i \preceq t'_i$ and $t_j \preceq t'_j$), if $l_i
  \preceq_n l_j$ holds (i.e., $t_i \preceq t_j$ and $t'_i \preceq
  t'_j$ hold), then $t_i \preceq t'_j$ holds (because both $t_i
  \preceq t_j \preceq t'_j$ and $t_i \preceq t'_i \preceq t'_j$ hold,
  and relation $\preceq$ is transitive), but $t'_i \preceq t_j$ does
  not (because of inconclusive $t'_i \succeq t_i \preceq t_j$ and
  $t'_i \preceq t'_j \succeq t_j$), and therefore $l_i \preceq_x l_j$
  does not hold.
\end{proof}

\end{document}

%% file: fig_graph.tikz
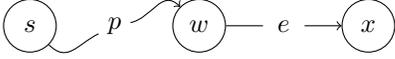
\begin{figure}
  \centering%
  \begin{tikzpicture}[node distance = 2.25 cm, ->]

    \node[node] (s) {$s$};

    \node[node, right of = s] (w) {$w$};

    \node at ($(s)!0.5!(w)$) (p) {$p$};

    \node[node, right of = w] (x) {$x$};

    \draw (w) -- node [circle, fill = white] {$e$} (x);

    \draw
    (s) .. controls ($(s) + (0.5, -0.5)$) and ($(p) + (-0.5, -0.25)$) ..
    (p) .. controls ($(p) + (0.5, 0.1)$) and ($(w) + (-0.5, 0.5)$)..
    (w);

  \end{tikzpicture}
  \caption{Appending edge $e$ to path $p$.}%
  \label{f:graph}
\end{figure}

%% file: fig_nets.tikz
\savebox{\imageboxa}{
  \begin{tikzpicture}[node distance = 2.25 cm, ->]

    \node[node] (ns) {$n_s$};
    \coordinate[right of = ns] (u);
    \node[node] at ($(u) + (0, 0.75)$) (na) {$n_a$};
    \node[node] at ($(u) + (0, -0.75)$) (nb) {$n_b$};
    \node [node, right of = na] (nk) {$n_k$};

    \coordinate (ma) at ($ (ns)!.5!(na) $);
    \node at ($ (ma)!0.5 cm!90:(na) $) (ra) {$r_a$};
    \coordinate (mb) at ($ (ns)!.5!(nb) $);
    \node at ($ (mb)!0.5 cm!-90:(nb) $) (rb) {$r_b$};

    \draw (na) -- node [circle, fill = white] {$k$} (nk);

    \draw
    (ns) .. controls ($(ns) + (0, 1)$) and ($(ra) - (0.5, 0.2)$) ..
    (ra) .. controls ($(ra) + (0.5, 0.2)$) and ($(na) + (-0.5, 0.25)$)..
    (na);

    \draw
    (ns) .. controls ($(ns) + (0.5, -0.5)$) and ($(rb) - (1, 0.5)$) ..
    (rb) .. controls ($(rb) + (0.5, 0.5)$) and ($(nb) - (0.5, -0.5)$)..
    (nb);
  \end{tikzpicture}
}

\savebox{\imageboxb}{
  \begin{tikzpicture}[node distance = 2.25 cm, ->]

    \node[node] (ns) {$n_s$};
    \coordinate[right of = ns] (u);
    \node[node] at ($(u) + (0, 0.75)$) (na) {$n_a$};
    \node[node] at ($(u) + (0, -0.75)$) (nb) {$n_b$};

    \coordinate (ma) at ($ (ns)!.5!(na) $);
    \node at ($ (ma)!0.5 cm!90:(na) $) (ra) {$r_a$};
    \coordinate (mb) at ($ (ns)!.5!(nb) $);
    \node at ($ (mb)!0.5 cm!-90:(nb) $) (rb) {$r_b$};

    \draw (na) edge [bend left = 60] node [circle, fill = white] {$k$} (nb);

    \draw
    (ns) .. controls ($(ns) + (0, 1)$) and ($(ra) - (0.5, 0.2)$) ..
    (ra) .. controls ($(ra) + (0.5, 0.2)$) and ($(na) + (-0.5, 0.25)$)..
    (na);

    \draw
    (ns) .. controls ($(ns) + (0.5, -0.5)$) and ($(rb) - (1, 0.5)$) ..
    (rb) .. controls ($(rb) + (0.5, 0.5)$) and ($(nb) - (0.5, -0.5)$)..
    (nb);
  \end{tikzpicture}
}

\savebox{\imageboxc}{
  \begin{tikzpicture}[node distance = 2.25 cm, ->]

    \node[node] (ns) {$n_s$};

    \node[node, right of = ns] (na) {$n_a$};

    \coordinate (mid) at ($(ns)!0.5!(na)$);
    \node at ($(mid) + (0, 0.75)$) (ra) {$r_a$};
    \node at ($(mid) + (0, -0.75)$) (r'a) {$r'_a$};

    \node [node, right of = na] (nk) {$n_k$};

    \draw (na) -- node [circle, fill = white] {$k$} (nk);

    \draw
    (ns) .. controls ($(ns) + (0, 0.75)$) and ($(ra) + (-0.5, 0)$) ..
    (ra) .. controls ($(ra) + (0.75, 0.25)$) and ($(na) + (-0.5, 0.5)$)..
    (na);

    \draw
    (ns) .. controls ($(ns) + (0.5, -0.5)$) and ($(r'a) - (0.75, 0.25)$) ..
    (r'a) .. controls ($(r'a) + (0.5, 0.5)$) and ($(na) - (0.5, 0.75)$)..
    (na);
  \end{tikzpicture}
}

\begin{figure*}
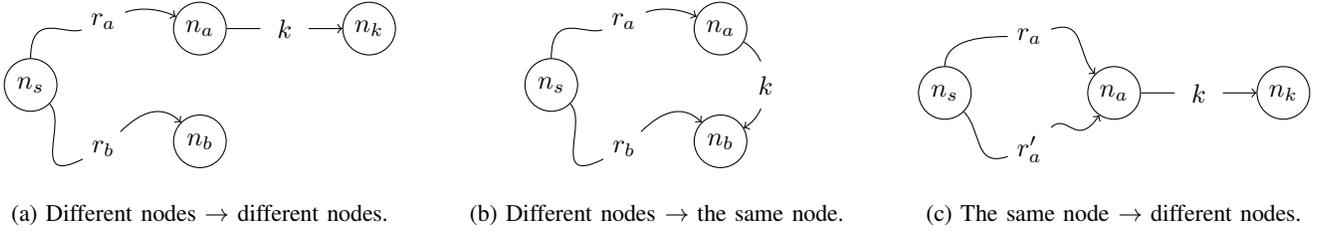

  \begin{subfigure}[b]{0.33\textwidth}%
    \centering\usebox{\imageboxa}%
    \caption{Different nodes $\to$ different nodes.}%
    \label{f:network1}%
  \end{subfigure}%
  \hfill%
  \begin{subfigure}[b]{0.33\textwidth}%
    \centering\usebox{\imageboxb}%
    \caption{Different nodes $\to$ the same node.}%
    \label{f:network2}%
  \end{subfigure}%
  \hfill%
  \begin{subfigure}[b]{0.33\textwidth}%
    \centering\raisebox{\dimexpr.5\ht\imageboxa-.5\height}
                       {\usebox{\imageboxc}}%
    \caption{The same node $\to$ different nodes.}%
    \label{f:network3}%
  \end{subfigure}%

  \caption{The only three possibilities of appending link $k$ to one
    of the routes in the network.}
  \label{f:nets}
\end{figure*}

%% file: tab_rels.tab

\begin{table}
  \caption{Relations of $l_i = (t_{i, a}, t_{i, b})$ and $l_j = (t_{j,
      a}, t_{j, b})$.}
  \label{t:l_rels}
  \centering%
  \begin{tabular}{c|c|c|c|c|}
    \cline{2-5}
    & $t_{i, b} \prec t_{j, b}$
    & $t_{i, b} = t_{j, b}$
    & $t_{i, b} \succ t_{j, b}$
    & $t_{i, b} \parallel t_{j, b}$\\
    \hline
    \multicolumn{1}{|c|}{$t_{i, a} \prec t_{j, a}$}
    & $l_i \prec_{\ne} l_j$
    & $l_i \prec_{\ne} l_j$
    & $l_i \parallel l_j$
    & $l_i \parallel l_j$\\
    \hline
    \multicolumn{1}{|c|}{$t_{i, a} = t_{j, a}$}
    & $l_i \prec_{\ne} l_j$
    & $l_i = l_j$
    & $l_i \succ_{\ne} l_j$
    & $l_i \parallel l_j$\\
    \hline
    \multicolumn{1}{|c|}{$t_{i, a} \succ t_{j, a}$}
    & $l_i \parallel l_j$
    & $l_i \succ_{\ne} l_j$
    & $l_i \succ_{\ne} l_j$
    & $l_i \parallel l_j$\\
    \hline
    \multicolumn{1}{|c|}{$t_{i, a} \parallel t_{j, a}$}
    & $l_i \parallel l_j$
    & $l_i \parallel l_j$
    & $l_i \parallel l_j$
    & $l_i \parallel l_j$\\
    \hline
  \end{tabular}
\end{table}

%% file: fig_substructure.tikz
\begin{figure}
  \centering
  \begin{tikzpicture}[node distance = 2.25 cm, ->]
    \node [node] (s) {$s$};
    \node [node, right of = s] (w) {$w$};

    \coordinate (mid) at ($(s)!0.5!(w)$);
    \node at ($(mid) + (0, 0.75)$) (pi) {$p_i$};
    \node at ($(mid) + (0, -0.75)$) (pj) {$p_j$};

    \node [node, right of = w] (x) {$x$};

    \draw
    (s) .. controls ($(s) + (-0.25, 1)$) and ($(pi) + (-0.5, -0.25)$) ..
    (pi) .. controls ($(pi) + (0.5, 0)$) and ($(w) + (-0.2, 0.75)$)..
    (w);

    \coordinate (pja) at ($(s) + (0.25, -0.5)$);
    \draw
    (s) .. controls ($(s) + (-0.1, -0.25)$) and ($(pja) + (-0.75, -0.75)$) ..
    (pja) .. controls ($(pja) + (0.25, 0.25)$) and ($(pj) + (-0.5, 0)$) ..
    (pj) .. controls ($(pj) + (1, 0.1)$) and ($(w) + (-1, -0.25)$)..
    (w);

    \draw (w) -- node [circle, fill = white] {$e$} (x);

  \end{tikzpicture}
  \caption{Efficient substructure.}
  \label{f:substructure}
\end{figure}
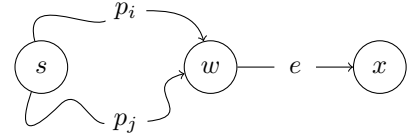

%% file: fig_lobe.tikz
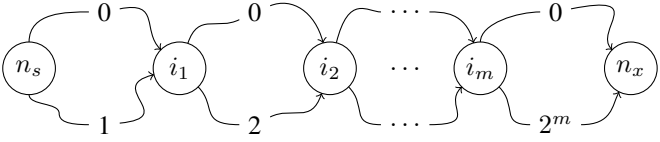
\begin{figure}
  \centering
  \begin{tikzpicture}[node distance = 2 cm, ->,
      outer sep = 0]
    \node [node] (ns) {$n_s$};
    \node [node, right of = s] (n1) {$i_1$};
    \node [node, right of = n1] (n2) {$i_2$};
    \node [node, right of = n2] (nm) {$i_m$};
    \node [node, right of = nm] (nx) {$n_x$};

    \coordinate (m1) at ($(ns)!0.5!(n1)$);
    \node at ($(m1) + (0, 0.75)$) (m1u) {0};
    \node at ($(m1) + (0, -0.75)$) (m1l) {1};

    \draw
    (ns) .. controls ($(ns) + (0, 0.75)$) and ($(m1u) + (-0.5, 0)$) ..
    (m1u) .. controls ($(m1u) + (0.75, 0.15)$) and ($(n1) + (-0.5, 0.5)$)..
    (n1);

    \coordinate (sx) at ($(ns) + (0.4, -0.6)$);
    \draw
    (ns) .. controls ($(ns) + (0, -0.5)$) and ($(sx) + (0, 0.15)$) ..
    (sx) .. controls ($(sx) + (0, -0.1)$) and ($(m1l) + (-0.5, 0)$) ..
    (m1l) .. controls ($(m1l) + (1, 0.1)$) and ($(n1) + (-1, -0.25)$)..
    (n1);

    \coordinate (m2) at ($(n1)!0.5!(n2)$);
    \node at ($(m2) + (0, 0.75)$) (m2u) {0};
    \node at ($(m2) + (0, -0.75)$) (m2l) {2};

    \draw
    (n1) .. controls ($(n1) + (0.25, 0.75)$) and ($(m2u) + (-0.5, -0.25)$) ..
    (m2u) .. controls ($(m2u) + (0.5, 0.2)$) and ($(n2) + (-0.2, 0.75)$)..
    (n2);

    \draw
    (n1) .. controls ($(n1) + (0.5, -0.5)$) and ($(m2l) - (0.75, 0)$) ..
    (m2l) .. controls ($(m2l) + (0.5, 0.25)$) and ($(n2) - (0.25, 0.75)$)..
    (n2);

    \node at ($(n2)!0.5!(nm)$) (mm) {$\ldots$};
    \node at ($(mm) + (0, 0.75)$) (mmu) {$\ldots$};
    \node at ($(mm) + (0, -0.75)$) (mml) {$\ldots$};

    \draw
    (n2) .. controls ($(n2) + (0.25, 1)$) and ($(mmu) + (-0.5, 0)$) ..
    (mmu) .. controls ($(mmu) + (0.5, 0)$) and ($(nm) + (-0.2, 0.75)$)..
    (nm);

    \draw
    (n2) .. controls ($(n2) + (0.5, -0.5)$) and ($(mml) - (0.75, 0)$) ..
    (mml) .. controls ($(mml) + (1, 0)$) and ($(nm) - (0.5, 0.5)$)..
    (nm);

    \coordinate (mt) at ($(nm)!0.5!(nx)$);
    \node at ($(mt) + (0, 0.75)$) (mtu) {0};
    \node at ($(mt) + (0, -0.75)$) (mtl) {$2^m$};

    \draw
    (nm) .. controls ($(nm) + (0, 0.6)$) and ($(mtu) + (-0.5, 0)$) ..
    (mtu) .. controls ($(mtu) + (1.25, 0.25)$) and ($(nx) + (-0.5, 0.5)$) ..
    (nx);

    \draw
    (nm) .. controls ($(nm) + (0.5, -0.5)$) and ($(mtl) - (0.75, 0)$) ..
    (mtl) .. controls ($(mtl) + (0.75, 0)$) and ($(nx) - (0.25, 0.5)$)..
    (nx);

  \end{tikzpicture}
  \caption{The worst case of exponential complexity.}
  \label{f:lobe}
\end{figure}

%% file: fig_example_network.tikz
\begin{figure}
  \centering
  \begin{tikzpicture}[->, node distance = 2.7 cm, outer sep = 0]
    \node [node] (n1) {$n_1$};
    \node [node, right of = n1] (n2) {$n_2$};
    \node [node, right of = n2] (n3) {$n_3$};

    \begin{scope}[every node/.style = {circle, fill = white,
          align = center, inner sep = 0}]
      \draw (n1) edge [bend left = 60]
      node {$k_1$\\$0, \RIO{0}{9}$} (n2);

      \draw (n1) edge [bend left = -60]
      node {$k_2$\\$1, \RIO{0}{9}$} (n2);

      \draw (n2) edge [bend left = 60]
      node {$k_3$\\$0, \RIO{5}{7}$} (n3);

      \draw (n2) edge [bend left = -60]
      node {$k_4$\\$2, \RIO{0}{2}$} (n3);
    \end{scope}

  \end{tikzpicture}
  \caption{The example network.}
  \label{f:en}
\end{figure}
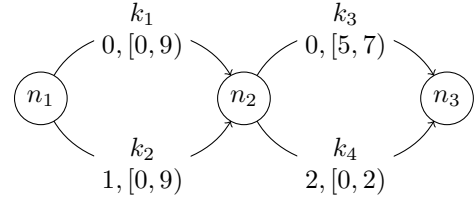

%% file: fig_example_graph.tikz

\newcommand{\edgelabel}[2]{#1, #2}

\begin{figure}
  \centering
  \begin{tikzpicture}[->, > = stealth]


    \begin{scope}[every node/.style = {draw, circle},
        node distance = \gd, on grid]

      \node (n11) {$(n_1, n_1)$};
      \node [right of = n11] (n12) {$(n_1, n_2)$};
      \node [right of = n12] (n13) {$(n_1, n_3)$};
      \node [below of = n12] (n22) {$(n_2, n_2)$};
      \node [below of = n13] (n23) {$(n_2, n_3)$};
      \node [below of = n23] (n33) {$(n_3, n_3)$};

    \end{scope}


    \begin{scope}[every node/.style = {rectangle, fill = white,
          align = center}]
      \begin{scope}[every edge/.style = {draw, bend left = 30}]
        \draw (n11) edge node {\edgelabel{$e_1$}{$k_1$}} (n12);
        \draw (n12) edge node {\edgelabel{$e_3$}{$k_3$}} (n13);
        \draw (n22) edge node {\edgelabel{$e_9$}{$k_3$}} (n23);
        \draw (n12) edge node {\edgelabel{$e_6$}{$k_2$}} (n22);
        \draw (n13) edge node {\edgelabel{$e_8$}{$k_2$}} (n23);
        \draw (n23) edge node {\edgelabel{$e_{12}$}{$k_4$}} (n33);
      \end{scope}

      \begin{scope}[every edge/.style = {draw, bend left = -30}]
        \draw (n11) edge node {\edgelabel{$e_2$}{$k_2$}} (n12);
        \draw (n12) edge node {\edgelabel{$e_4$}{$k_4$}} (n13);
        \draw (n22) edge node {\edgelabel{$e_{10}$}{$k_4$}} (n23);
        \draw (n12) edge node {\edgelabel{$e_5$}{$k_1$}} (n22);
        \draw (n13) edge node {\edgelabel{$e_7$}{$k_1$}} (n23);
        \draw (n23) edge node {\edgelabel{$e_{11}$}{$k_3$}} (n33);
      \end{scope}
    \end{scope}

  \end{tikzpicture}
  \caption{The example graph.}
  \label{f:eg}
\end{figure}
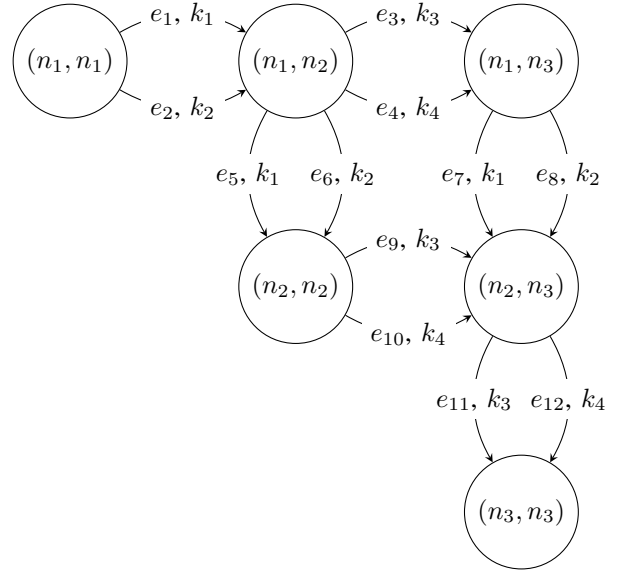

%% file: tab_labels1.tab

\begin{table}
  \footnotesize\centering
  \begin{tabular}{|l|l|l|l|}
    \hline
    Label &
    \begin{tabular}{@{}l@{}}For\\vertex\end{tabular} &
    \begin{tabular}{@{}l@{}}Yielded\\by\end{tabular} &
    \begin{tabular}{@{}l@{}}Dropped\\because\end{tabular}\\
    \hline
    $l_1 = ((0, \Omega), (0, \Omega))$ & $(n_1, n_1)$ & nothing &\\
    $l_2 = ((0, \Omega), (0, \RIO{0}{9}))$ & $(n_1, n_2)$ & $l_1 \oplus e_1$ &\\
    $l_3 = ((0, \Omega), (1, \RIO{0}{9}))$ & $(n_1, n_2)$ & $l_1 \oplus e_2$ & $l_2 \preceq l_3$\\
    $l_4 = ((0, \Omega), (0, \RIO{5}{7}))$ & $(n_1, n_3)$ & $l_2 \oplus e_3$ &\\
    $l_5 = ((0, \Omega), (2, \RIO{0}{2}))$ & $(n_1, n_3)$ & $l_2 \oplus e_4$ &\\
    $l_6 = ((0, \RIO{0}{9}), (0, \RIO{0}{9}))$ & $(n_2, n_2)$ & $l_2 \oplus e_5$ & $k_1$ reused\\
    $l_7 = ((0, \RIO{0}{9}), (1, \RIO{0}{9}))$ & $(n_2, n_2)$ & $l_2 \oplus e_6$ &\\
    $l_8 = ((0, \RIO{0}{9}), (0, \RIO{5}{7}))$ & $(n_2, n_3)$ & $l_4 \oplus e_7$ & $k_1$ reused\\
    $l_9 = ((1, \RIO{0}{9}), (0, \RIO{5}{7}))$ & $(n_2, n_3)$ & $l_4 \oplus e_8$ &\\
    $l_{10} = ((1, \RIO{0}{9}), (0, \RIO{5}{7}))$ & $(n_2, n_3)$ & $l_7 \oplus e_9$ & $l_9 \preceq l_{10}$\\
    $l_{11} = ((0, \RIO{0}{9}), (1, \RIO{5}{7}))$ & $(n_2, n_3)$ & $l_7 \oplus e_9$ &\\
    $l_{12} = ((1, \RIO{0}{9}), (2, \RIO{0}{2}))$ & $(n_2, n_3)$ & $l_7 \oplus e_{10}$ &\\
    $l_{13} = ((0, \RIO{0}{9}), (3, \RIO{0}{2}))$ & $(n_2, n_3)$ & $l_7 \oplus e_{10}$ &\\
    $l_{14} = ((0, \RIO{5}{7}), (1, \RIO{5}{7}))$ & $(n_3, n_3)$ & $l_{11} \oplus e_{11}$ & $k_3$ reused\\
    $l_{15} = ((2, \RIO{0}{2}), (1, \RIO{5}{7}))$ & $(n_3, n_3)$ & $l_{11} \oplus e_{12}$ &\\
    $l_{16} = ((0, \RIO{5}{7}), (1, \RIO{5}{7}))$ & $(n_3, n_3)$ & $l_9 \oplus e_{11}$ & $k_3$ reused\\
    $l_{17} = ((0, \RIO{5}{7}), (3, \RIO{0}{2}))$ & $(n_3, n_3)$ & $l_9 \oplus e_{12}$ &\\
    $l_{18} = ((0, \RIO{0}{9}), (2, \RIO{0}{2}))$ & $(n_2, n_3)$ & $l_5 \oplus e_7$ & $k_1$ reused\\
    $l_{19} = ((1, \RIO{0}{9}), (2, \RIO{0}{2}))$ & $(n_2, n_3)$ & $l_5 \oplus e_8$ & $l_{12} \preceq l_{19}$\\
    $l_{20} = ((1, \RIO{5}{7}), (2, \RIO{0}{2}))$ & $(n_3, n_3)$ & $l_{12} \oplus e_{11}$ & $l_{15} \preceq l_{20}$\\
    $l_{21} = ((3, \RIO{0}{2}), (2, \RIO{0}{2}))$ & $(n_3, n_3)$ & $l_{12} \oplus e_{12}$ & $k_4$ reused\\
    $l_{22} = ((0, \RIO{5}{7}), (3, \RIO{0}{2}))$ & $(n_3, n_3)$ & $l_{13} \oplus e_{11}$ & $l_{17} \preceq l_{22}$\\
    $l_{23} = ((2, \RIO{0}{2}), (3, \RIO{0}{2}))$ & $(n_3, n_3)$ & $l_{13} \oplus e_{12}$ & $k_4$ reused\\
    \hline
  \end{tabular}
  \caption{Labels of the example for $\preceq$.}
  \label{t:el1}
\end{table}

%% file: fig_example_tree1.tikz

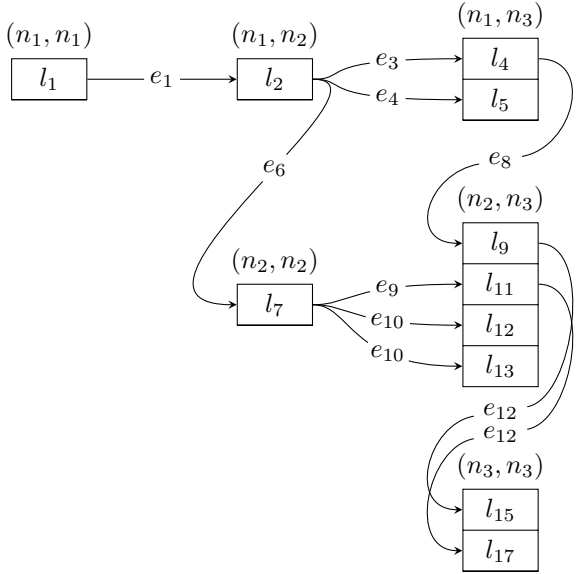
\begin{figure}
  \centering
  \begin{tikzpicture}[> = stealth]

    \begin{scope}[every node/.style = {rectangle split, draw,
          rectangle split ignore empty parts, anchor = north,
          minimum width = 1 cm}, node distance = \gd]

      \node (n11) {$l_1$};
      \node [right of = n11] (n12) {$l_2$};
      \node [right of = n12] (n13) {$l_4$\nodepart{two}$l_5$};
      \node [below of = n12] (n22) {$l_7$};
      \node [below of = n13] (n23) {$l_9$\nodepart{two}
        $l_{11}$\nodepart{three}$l_{12}$\nodepart{four}$l_{13}$};
      \node [below of = n23] (n33) {$l_{15}$\nodepart{two}$l_{17}$};

    \end{scope}

    \node[above = \dan of n11.north] {$(n_1, n_1)$};
    \node[above = \dan of n12.north] (n12l) {$(n_1, n_2)$};
    \node[above = \dan of n13.north] (n13l) {$(n_1, n_3)$};
    \node[above = \dan of n22.north] {$(n_2, n_2)$};
    \node[above = \dan of n23.north] (n23l) {$(n_2, n_3)$};
    \node[above = \dan of n33.north] (n33l) {$(n_3, n_3)$};


    \node (n13n23a) at ($(n13.south)!0.5!(n23l)$) {$e_8$};

    \coordinate (n23n33) at ($(n23.south)!0.45!(n33l)$);
    \node (n23n33a) at ($(n23n33) + (0, 0.15)$) {$e_{12}$};
    \node (n23n33b) at ($(n23n33) - (0, 0.15)$) {$e_{12}$};

    \begin{scope}[->, every node/.style = {circle, fill = white,
          inner sep = 1 pt}]

      \draw (n11.text east) to [out = 0, in = 180]
      node {$e_1$} (n12.text west);

      \path [name path = ver23] ($(n12l)!0.5!(n13l)$)
      -- ($(n22)!0.5!(n33)$);

      \draw [name path = n12n13a] (n12.text east)
      to [controls = +(0:0.5) and +(180:2)] (n13.text west);
      \path [name intersections = {of = n12n13a and ver23}]
        node at (intersection-1) {$e_3$};

      \draw [name path = n12n13b] (n12.text east)
      to [controls = +(0:0.5) and +(180:2)] (n13.two west);
      \path [name intersections = {of = n12n13b and ver23}]
        node at (intersection-1) {$e_4$};
      
      \draw [name path = n12n22] (n12.text east)
      to [controls = +(0:1.25) and +(180:2)] (n22.text west);
      \path [name path = n12--n22] (n12) -- (n22);
      \path [name intersections = {of = n12n22 and n12--n22}]
      node at (intersection-1) {$e_6$};

      \draw[-] (n13.text east) to [controls = +(0:0.75) and +(10:1)]
      (n13n23a);
      \draw (n13n23a) to [controls = +(190:1) and +(180:0.75)]
      (n23.text west);

      \draw [name path = n22n23a] (n22.text east)
      to [controls = +(0:0.5) and +(180:1.5)] (n23.two west);
      \path [name intersections = {of = n22n23a and ver23}]
      node at (intersection-1) {$e_9$};

      \draw [name path = n22n23b] (n22.text east)
      to [controls = +(0:0.5) and +(180:1.5)] (n23.three west);
      \path [name intersections = {of = n22n23b and ver23}]
      node at (intersection-1) {$e_{10}$};

      \draw [name path = n22n23c] (n22.text east)
      to [controls = +(0:0.5) and +(180:1.5)] (n23.four west);
      \path [name intersections = {of = n22n23c and ver23}]
      node at (intersection-1) {$e_{10}$};

      \draw[-] (n23.text east) to [controls = +(0:0.75) and +(10:1)]
      (n23n33a);
      \draw (n23n33a) to [controls = +(190:1) and +(180:0.75)]
      (n33.text west);
      \draw[-] (n23.two east) to [controls = +(0:0.75) and +(10:1)]
      (n23n33b);
      \draw (n23n33b) to [controls = +(190:1) and +(180:0.75)]
      (n33.two west);

    \end{scope}

  \end{tikzpicture}
  \caption{The $\preceq$-efficient-label tree.}
  \label{f:et1}
\end{figure}

%% file: tab_labels2.tab

\begin{table}
  \footnotesize\centering
  \begin{tabular}{|l|l|l|l|}
    \hline
    Label &
    \begin{tabular}{@{}l@{}}For\\vertex\end{tabular} &
    \begin{tabular}{@{}l@{}}Yielded\\by\end{tabular} &
    \begin{tabular}{@{}l@{}}Dropped\\because\end{tabular}\\
    \hline
    $l_1 = ((0, \Omega), (0, \Omega))$ & $(n_1, n_1)$ & nothing &\\
    $l_2 = ((0, \Omega), (0, \RIO{0}{9}))$ & $(n_1, n_2)$ & $l_1 \oplus e_1$ &\\
    $l_3 = ((0, \Omega), (1, \RIO{0}{9}))$ & $(n_1, n_2)$ & $l_1 \oplus e_2$ & $l_2 \preceq' l_3$\\
    $l_4 = ((0, \Omega), (0, \RIO{5}{7}))$ & $(n_1, n_3)$ & $l_2 \oplus e_3$ &\\
    $l_5 = ((0, \Omega), (2, \RIO{0}{2}))$ & $(n_1, n_3)$ & $l_2 \oplus e_4$ &\\
    $l_6 = ((0, \RIO{0}{9}), (0, \RIO{0}{9}))$ & $(n_2, n_2)$ & $l_2 \oplus e_5$ & $k_1$ reused\\
    $l_7 = ((0, \RIO{0}{9}), (1, \RIO{0}{9}))$ & $(n_2, n_2)$ & $l_2 \oplus e_6$ &\\
    $l_8 = ((0, \RIO{0}{9}), (0, \RIO{5}{7}))$ & $(n_2, n_3)$ & $l_4 \oplus e_7$ & $k_1$ reused\\
    $l_9 = ((1, \RIO{0}{9}), (0, \RIO{5}{7}))$ & $(n_2, n_3)$ & $l_4 \oplus e_8$ &\\
    $l_{10} = ((1, \RIO{0}{9}), (0, \RIO{5}{7}))$ & $(n_2, n_3)$ & $l_7 \oplus e_9$ & $l_9 \preceq' l_{10}$\\
    $l_{11} = ((0, \RIO{0}{9}), (1, \RIO{5}{7}))$ & $(n_2, n_3)$ & $l_7 \oplus e_9$ & $l_9 \preceq' l_{11}$\\
    $l_{12} = ((1, \RIO{0}{9}), (2, \RIO{0}{2}))$ & $(n_2, n_3)$ & $l_7 \oplus e_{10}$ &\\
    $l_{13} = ((0, \RIO{0}{9}), (3, \RIO{0}{2}))$ & $(n_2, n_3)$ & $l_7 \oplus e_{10}$ & $l_{12} \preceq' l_{13}$\\
    $l_{16} = ((0, \RIO{5}{7}), (1, \RIO{5}{7}))$ & $(n_3, n_3)$ & $l_9 \oplus e_{11}$ & $k_3$ reused\\
    $l_{17} = ((0, \RIO{5}{7}), (3, \RIO{0}{2}))$ & $(n_3, n_3)$ & $l_9 \oplus e_{12}$ &\\
    $l_{18} = ((0, \RIO{0}{9}), (2, \RIO{0}{2}))$ & $(n_2, n_3)$ & $l_5 \oplus e_7$ & $k_1$ reused\\
    $l_{19} = ((1, \RIO{0}{9}), (2, \RIO{0}{2}))$ & $(n_2, n_3)$ & $l_5 \oplus e_8$ & $l_{12} \preceq' l_{19}$\\
    $l_{20} = ((1, \RIO{5}{7}), (2, \RIO{0}{2}))$ & $(n_3, n_3)$ & $l_{12} \oplus e_{11}$ & $l_{17} \preceq' l_{20}$\\
    $l_{21} = ((3, \RIO{0}{2}), (2, \RIO{0}{2}))$ & $(n_3, n_3)$ & $l_{12} \oplus e_{12}$ & $k_4$ reused\\
    \hline
  \end{tabular}
  \caption{Labels of the example for $\preceq'$.}
  \label{t:el2}
\end{table}

%% file: fig_example_tree2.tikz

\begin{figure}
  \centering
  \begin{tikzpicture}[> = stealth]

    \begin{scope}[every node/.style = {rectangle split, draw,
          rectangle split ignore empty parts, anchor = north,
          minimum width = 1 cm}, node distance = \gd]

      \node (n11) {$l_1$};
      \node [right of = n11] (n12) {$l_2$};
      \node [right of = n12] (n13) {$l_4$\nodepart{two}$l_5$};
      \node [below of = n12] (n22) {$l_7$};
      \node [below of = n13] (n23) {$l_9$\nodepart{two}$l_{12}$};
      \node [below of = n23] (n33) {$l_{17}$};

    \end{scope}

    \node[above = \dan of n11.north] {$(n_1, n_1)$};
    \node[above = \dan of n12.north] (n12l) {$(n_1, n_2)$};
    \node[above = \dan of n13.north] (n13l) {$(n_1, n_3)$};
    \node[above = \dan of n22.north] {$(n_2, n_2)$};
    \node[above = \dan of n23.north] (n23l) {$(n_2, n_3)$};
    \node[above = \dan of n33.north] (n33l) {$(n_3, n_3)$};


    \node (n13n23a) at ($(n13.south)!0.5!(n23l)$) {$e_8$};

    \node (n23n33) at ($(n23.south)!0.45!(n33l)$) {$e_{12}$};

    \begin{scope}[->, every node/.style = {circle, fill = white,
          inner sep = 1 pt}]

      \draw (n11.text east) to [out = 0, in = 180]
      node {$e_1$} (n12.text west);

      \path [name path = ver23] ($(n12l)!0.5!(n13l)$)
      -- ($(n22)!0.5!(n33)$);

      \draw [name path = n12n13a] (n12.text east)
      to [controls = +(0:0.5) and +(180:2)] (n13.text west);
      \path [name intersections = {of = n12n13a and ver23}]
        node at (intersection-1) {$e_3$};

      \draw [name path = n12n13b] (n12.text east)
      to [controls = +(0:0.5) and +(180:2)] (n13.two west);
      \path [name intersections = {of = n12n13b and ver23}]
        node at (intersection-1) {$e_4$};
      
      \draw [name path = n12n22] (n12.text east)
      to [controls = +(0:1.25) and +(180:2)] (n22.text west);
      \path [name path = n12--n22] (n12) -- (n22);
      \path [name intersections = {of = n12n22 and n12--n22}]
      node at (intersection-1) {$e_6$};

      \draw[-] (n13.text east) to [controls = +(0:0.75) and +(10:1)]
      (n13n23a);
      \draw (n13n23a) to [controls = +(190:1) and +(180:0.75)]
      (n23.text west);

      \draw [name path = n22n23b] (n22.text east)
      to [controls = +(0:0.5) and +(180:1.5)] (n23.two west);
      \path [name intersections = {of = n22n23b and ver23}]
      node at (intersection-1) {$e_{10}$};

      \draw[-] (n23.text east) to [controls = +(0:0.75) and +(10:1)]
      (n23n33);
      \draw (n23n33a) to [controls = +(190:1) and +(180:0.75)]
      (n33.text west);

    \end{scope}

  \end{tikzpicture}
  \caption{The $\preceq'$-efficient-label tree.}
  \label{f:et2}
\end{figure}
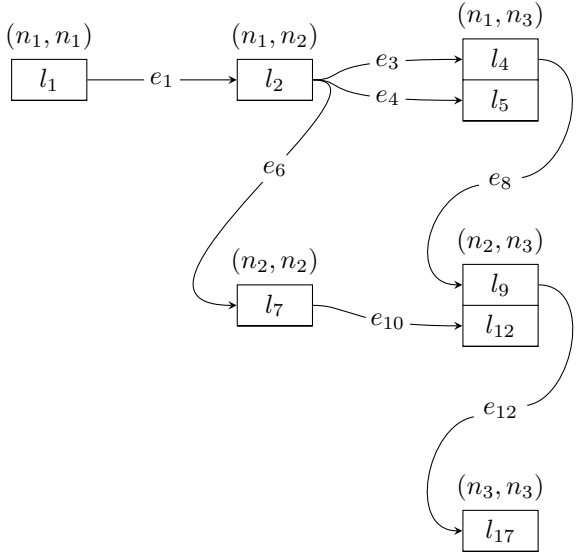